\documentclass[submission,copyright,creativecommons]{eptcs}
\providecommand{\event}{SOS 2007} 

\usepackage{iftex}
\usepackage{todonotes}

\ifpdf
  \usepackage{underscore}         
  \usepackage[T1]{fontenc}        
\else
  \usepackage{breakurl}           
\fi

\usepackage{enumerate,xspace}
\usepackage{amsmath,amssymb,wasysym}
\usepackage[all]{xy}
\usepackage{multicol}
\usepackage{proof}
\usepackage{tikz-cd}

\newtheorem{observation}{Remark}[section]
\newtheorem{lemma}[observation]{Lemma}  
\newtheorem{theorem}[observation]{Theorem}
\newtheorem{definition}[observation]{Definition}
\newtheorem{example}[observation]{Example}
\newtheorem{remark}[observation]{Remark}

\newtheorem{proposition}[observation]{Proposition} 
\newtheorem{corollary}[observation]{Corollary}

\title{Generalized Inverses of Quantum Channels: \\a categorical perspective}
\author{Robin Cockett\thanks{Partially funded by NSERC.}
\institute{Department of Computer Science\\
University of Calgary, Canada}
\email{robin@ucalgary.ca}
\and
Jean-Simon Pacaud Lemay
\institute{School of Mathematical and Physical Sciences\\
Macquarie University, Australia}
\email{js.lemay@mq.edu.au}
\and
Priyaa Varshinee Srinivasan
\thanks{This work was co-funded by the European
Union and Estonian Research Council through the Mobilitas 3.0 (MOB3JD1227).}
\institute{Department of Software Sciences\\
Tallinn University of Technology, Estonia}
\email{priyaavarshinee@gmail.com}
}

\def\event{QPL 2026}
\def\titlerunning{Generalized Inverses of Quantum Channels: a categorical perspective}
\def\authorrunning{J.R.B. Cockett, J.-S. P. Lemay, \& P. V. Srinivasan}

\begin{document}
\maketitle

\begin{abstract}
A quantum channel is defined as being completely positive (CP) and trace preserving (TP). While not every quantum channel is invertible or reversible, every quantum channel admits various kinds of generalized inverses such as the Moore-Penrose inverse and the Drazin inverse. A generalized inverse of a quantum channel may not itself be a quantum channel: it often fails to be CP.  However, generalized inverses still play an important role in quantum error mitigation.  Here, because it is often desirable for the generalized inverse of a quantum channel to be at least TP, the Drazin inverse, which is TP, is favoured over the Moore-Penrose inverse, which is not in general TP. 

In this paper, we take a categorical perspective on generalized inverses of quantum channels. This allows us to give a simple proof of the fact that the Drazin inverse of a quantum channel is always TP. It also allows us to show that for unital quantum channels, the Drazin inverse is also unital. We then generalize this result to dagger Drazin inverses, which allows us to show that for unital quantum channels, the Moore-Penrose inverse is both TP and unital as well. This opens the door to new applications of both the Drazin inverse and Moore-Penrose inverse in quantum information theory and, in particular, in quantum error mitigation. 

\end{abstract}

\section{Introduction}

In the era of noisy intermediate-scale quantum (NISQ) computing \cite{Pre18}, both quantum error correction \cite{chatterjee2023quantum} and quantum error mitigation \cite{cai2023quantum,cao2021nisq} are essential for reliable quantum computation. Generalized inverses \cite{campbell2009generalized} of quantum channels, i.e. completely positive (CP) trace-preserving (TP) linear maps, play an important role in quantum error mitigation. 

While quantum error correction protocols  aim to eliminate noise by applying an inverse quantum channel, quantum error mitigation protocols  reduce noise by post-processing measurement results using generalized inverses  when the noise  is not reversible. For quantum error mitigation protocols, the generalized inverses of quantum channels no longer need to be physically implementable, hence need not be CP. However, it is still desirable for the generalized inverse to be TP since, for example, it can provide better stability for computer simulations. For this reason, the Drazin inverse \cite[Chap 7]{campbell2009generalized} is favoured over the Moore-Penrose inverse \cite[Chap 1]{campbell2009generalized} for quantum error mitigation since the Drazin inverse of a TP map is again TP \cite[Thm 2]{cao2021nisq}, while the Moore-Penrose inverse neefd not be \cite[Ex 2]{cao2021nisq}.

Quoting \cite{cao2021nisq}: ``{\em To our knowledge, generalized non-CPTP inverses for non-invertible quantum channels have not been studied extensively in previous literature. The new understanding of the natural representation opens the possibility of studying the structures and properties of these maps from mathematical interests. Different constructions of generalized inverse maps bring new properties. This research direction can provide a guideline for implementing these maps in EM and also in other areas of quantum information sciences.}" Motivated by this philosophy, the purpose of this paper is to continue the study of generalized inverses of quantum channels from a categorical perspective, combining categorical quantum foundations and the study of generalized inverses in categories. Indeed, abstractions of quantum channels, especially the notion of being CP, are well-established and well-studied concepts in dagger compact closed categories \cite{coecke2016cpinf,coecke2016categories,Coecke2012,Cunningham2015,Gogioso2019,selinger2007dagger}. On the other hand, there has been recent interest on Moore-Penrose and Drazin inverses in dagger categories \cite{Cockett2023Moore,cockett2025dagger,cockett2024drazin,EPTCS426.3}. 

The main theoretical result of \cite{cao2021nisq} is that the Drazin inverse of a TP map is again TP, the proof \cite[App A]{cao2021nisq} involved  heavy duty linear algebra. The first main consequence of our categorical approach is a much more streamlined and shorter proof of the same result using two simple commuting triangles  (Prop \ref{prop:DrazinTP}), which are a consequence of the algebraic properties of Drazin inverses available in any category \cite[Prop 3.10]{cockett2024drazin}. This also allows us to show that for \textit{unital} quantum channels, the Drazin inverse is also unital (Lemma \ref{lemma:DrazinU}). 

One disadvantage of Drazin inverses is that they are only defined for endomaps. This limitation is removed by considering dagger Drazin inverses, a generalization of Drazin inverses in dagger categories \cite{cockett2025dagger}. We provide the analogue of \cite[Prop 3.10]{cockett2024drazin} for dagger Drazin inverses (Prop \ref{dagger-Drazin-commuting}), which then allows us to show that the dagger Drazin inverse of a TP and untial map is again TP and unital (Prop \ref{dag-Drazin-TP}). 

As mentioned above, one of the key motivating observations in \cite{cao2021nisq} was that while the Drazin inverse of a quantum channel is always TP, its Moore-Penrose inverse need not be. The second main result of our categorical approach is we prove that for unital quantum channels, the Moore-Penrose inverse is both TP and unital. In fact, we show that a map is TP and unital if and only if its Moore-Penrose inverse is TP and unital (Prop \ref{MP-TP}). This observation could potentially lead to new applications of Moore-Penrose inverses in quantum information theory, and, in particular, to applications of the Moore-Penrose inverse of unital quantum channels in quantum error mitigation. For example, mixed unitary channels are central to the study of noise and decoherence in quantum systems with wide-ranging applications in quantum cryptography \cite{ambainis2000private,ambainis2004small,hayden2004randomizing} and error mitigation \cite{TuJ25}. Since a mixed unitary channel is a unital quantum channel, as a direct application of our result, we get that its Moore-Penrose inverse is both TP and unital. As far as we know, this behaviour of mixed unitary channels has not been noted previously. 

\textbf{Acknowledgements:} The authors thank Chris Heunen, Masahito Hasegawa, Cole Comfort, Amar Hadzihasanovic, and Martti Karvonen for useful discussions and answering our multiple questions. 

\section{Quantum Channels in Dagger Compact Closed Categories}

In this section, we review the notion of abstract quantum channels in dagger compact closed categories. We assume that the reader is familiar with the basics of dagger monoidal categories. For a detailed introduction to dagger compact closed categories, we refer the reader to \cite{selinger2010survey} (which also includes the graphical calculus for dagger compact closed categories, which we will not make use of in this paper), and for a detailed introduction to complete positivity in this setting, we refer the reader to \cite{selinger2007dagger}. 

For an arbitrary category $\mathbb{X}$, we will denote objects using capital letters $A,B,C$, etc. and maps will be denoted as arrows $f: A \to B$, with homsets as $\mathbb{X}(A,B)$. Identity maps will be denoted as $1_A: A \to A$ (or simply $1$ when there is no confusion) while composition will be denoted diagrammatically, so the composite of $f: A \to B$ and $g: B \to C$ is written as $f;g: A \to C$, which is read as first do $f$ then $g$. For simplicity, we allow ourselves to work in a \textit{strict} monoidal category setting, meaning that the associativity and unital isomorphisms of the monoidal product are identities. So, for a symmetric (strict) monoidal category $\mathbb{X}$, we denotes its monoidal product as $\otimes$, its monoidal unit as $I$, and the natural symmetry isomorphism as $\sigma_{A,B}: A \otimes B \to B \otimes A$. Strictness allows us to write down $A_1 \otimes \hdots \otimes A_n$ and $A \otimes I = A = I \otimes A$. 

Recall that a \textbf{compact closed category} \cite[Sec 4.8]{selinger2010survey} is a symmetric monoidal category such that every object has a chosen dual. Explicitly, for every object $A$, there is an object $A^\ast$, called the \textbf{dual} of $A$, equipped with a pair of maps $\cup_A: A^\ast \otimes A \to I$, called the \textbf{cup}, and $\cap_A: I \to A \otimes A^\ast$, called the \textbf{cap}, such that the famous snake equations hold. For convenience, we will work with the convention that $A^{\ast\ast} = A$, $I^\ast = I$, and $(A \otimes B)^\ast = B^\ast \otimes A^\ast$. On the other hand, a \textbf{dagger category} \cite[Sec 7]{selinger2010survey} is a category $\mathbb{X}$ equipped with a contravariant functor $\dagger: \mathbb{X} \to \mathbb{X}$ which is the identity on objects and also involutive. Explicitly, a dagger associates every map $f: A \to B$ to a map of dual type $f^\dagger: B \to A$, called its \textbf{adjoint}, such that the following identities hold:
\begin{align*}
    (f;g)^\dagger= g^\dagger; f^\dagger && 1^\dagger = 1 && f^{\dagger \dagger} = f
\end{align*}
A \textbf{dagger symmetric monoidal category} \cite[Sec 7.2]{selinger2010survey} is a symmetric monoidal category which has a dagger which is compatible with the symmetric monoidal structure in the sense that $\otimes$ is a dagger functor, which amounts to asking that $(f \otimes g)^\dagger = f^\dagger \otimes g^\dagger$, and also that the adjoint of the natural symmetry isomorphism is its inverse, so we have $\sigma_{A,B}^\dagger = \sigma^{-1}_{A,B} = \sigma_{B,A}$ (in other words, $\sigma$ is unitary). A \textbf{dagger compact closed category} \cite[Sec 7.4]{selinger2010survey} (or $\dagger$KCC for short) is a compact closed category which is also a dagger symmetric monoidal category such that the adjoint of the cap (resp. cup) is the cup (resp. cap) up to a twist, that is, the following equalities hold: 
\begin{align*}
    \cap_A^\dagger =  \sigma_{A,A^\ast}; \cup_A && \cup_A^\dagger = \cap_A; \sigma_{A, A^\ast}
\end{align*}

\begin{example} Let $\mathbb{C}$ be the field of complex numbers, and let $\mathsf{FHILB}$ be the category of finite dimensional (complex) Hilbert spaces and $\mathbb{C}$-linear maps between them. Then $\mathsf{FHILB}$ is a $\dagger$KCC. The monoidal structure is given by the usual tensor product of Hilbert spaces, $H \otimes K$, and the monoidal unit is $\mathbb{C}$. The dual of $H$ is the space of linear functionals, $H^\ast = \mathsf{FHILB}(H, \mathbb{C})$. In bra-ket notation, if $\lbrace \vert e_i \rangle \rbrace^n_{i=1}$ is a basis of $H$, with $\lbrace \langle e_i \vert \rbrace^n_{i=1}$ the dual basis for $H^\ast$, the cup and cap are defined respectively as follows: 
\begin{align*}
    \cup_H\left(\langle e_i \vert \otimes \vert e_j \rangle \right) = \langle e_i \vert e_j \rangle && \cap_H(1) = \sum^n_{i=1} \vert e_i \rangle\langle e_i \vert
\end{align*}
The dagger is given by the adjoint, that is, for a $\mathbb{C}$-linear map $f: H \to K$, $f^\dagger: K \to H$ is the unique map such that the following equality holds: 
\begin{align*}
    \langle f(x), y \rangle = \langle x, f^\dagger(y)\rangle
\end{align*}
\end{example}

\begin{example} Let $\mathsf{REL}$ be the category of sets and relations, where recall the objects are sets and where a map $R: X \to Y$ is a subset $R \subseteq X \times Y$. Then $\mathsf{REL}$ is a $\dagger$KCC. The monoidal structure is given by the Cartesian product of sets, $X \times Y$, and the monoidal unit is a chosen singleton $\lbrace \ast \rbrace$. The dual of $X$ is itself, $X^\ast = X$, while the cup and cap are the subsets:
\begin{align*}
    \cap_X := \lbrace (\ast, x,x) \vert \forall x \in X \rbrace \subseteq \lbrace \ast \rbrace \times X \times X && \cup_X:= \lbrace (x,x,\ast) \vert \forall x \in X\rbrace \subseteq X \times X \times \lbrace \ast \rbrace
\end{align*}
The dagger is given by the converse relation, so for $R \subseteq X \times Y$, its adjoint is the subset:
\begin{align*}
    R^\dagger := \lbrace (y,x) \vert (x,y) \in R \rbrace \subseteq Y \times X
\end{align*}
\end{example}

Now before we turn our attention to reviewing quantum channels in a $\dagger$KCC, for the intuition of the story of this paper, we find it useful to consider $\dagger$KCC from the perspective of being \textit{closed} and having a \textit{trace}. This will help provide some useful notation that may be more familiar for certain readers. 

So briefly recall that a \textbf{symmetric monoidal closed category} is a symmetric monoidal category $\mathbb{X}$ such that for every pair of objects $A$ and $B$, there is an object $[A,B]$ called the \textbf{internal hom}, such that we have natural isomorphisms $\mathbb{X}(C \otimes A, B) \cong \mathbb{X}(C, [A,B])$. Intuitively we think of $[A,B]$ as an object representing all maps from $A$ to $B$. Now every compact closed category is also symmetric monoidal closed by setting $[A,B] = A^\ast \otimes B$. In fact compact closed categories are precisely the symmetric monoidal closed categories such that the canonical map $[A,I] \otimes B \to [A,B]$ is an isomorphism. Of particular importance for the story of quantum channels will be the object $[A,A]= A^\ast \otimes A$, which we interpret as the object representing the endomaps of type $A$. 

On the other hand, a \textbf{traced symmetric monoidal closed category} \cite{hasegawa2009traced} is a symmetric monoidal closed category equipped with a \textbf{trace operator} $\mathsf{Tr}$, which is a family of functions indexed by triples of objects, $\mathsf{Tr}^X_{A,B}: \mathbb{X}(X \otimes A, X \otimes B) \to \mathbb{X}(A,B)$, where for a map $f: X \otimes A \to X \otimes B$, the resulting map $\mathsf{Tr}^X_{A,B}(f): A \to B$ is called the \textbf{partial trace of $f$ at $X$}. This trace operator $\mathsf{Tr}$ also satisfies well-known axioms \cite[Sec 5]{selinger2010survey} such as naturality in $A$ and $B$, extranaturality in $X$, superposition, and yanking. Now in the special case that $A=B=I$, we will simply denote the trace operator as $\mathsf{Tr}^X = \mathsf{Tr}^X_{I,I}$, and in this case, for an endomap $f: X \to X$, we call $\mathsf{Tr}^X(f): I \to I$ its \textbf{trace}. In a traced symmetric monoidal closed category, we can internalize the trace operator to obtain maps of type $\mathsf{tr}^X_{A,B}: [X \otimes A, X \otimes B] \to [A,B]$ and $\mathsf{tr}^X = \mathsf{tr}^X_{I,I}: [X,X] \to I$ respectfully. 

Now every compact closed category is a traced symmetric monoidal closed category in a unique way \cite[Sec 3.2]{hasegawa2009traced}. In particular, the internal trace operators are given precisely by the cups:
\begin{align*}
    \mathsf{tr}^X_{A,B} = 1_{A^\ast} \otimes \cup_X \otimes 1_B && \mathsf{tr}^X= \cup_X
\end{align*}
Now if the cups gives the internal trace, what do the caps give? Caps instead give the internal way of picking out identity maps in the internal hom. Indeed, in a symmetric monoidal closed category, we have canonical maps $\mathsf{u}^X : I \to [X,X]$ and $\mathsf{u}^X_{A,B}: [A,B] \to [X \otimes A, X \otimes B]$, where the former is interpreted as picking out $1_X$ in $[X,X]$, while the latter sends $f: A \to B$ to $1_X \otimes f$. In a compact closed category, these maps are given by the cap (up to a twist): 
\begin{align*}
    \mathsf{u}^X_{A,B} =  (1_{A^\ast} \otimes \cap_X \otimes 1_B);(1_{A^\ast} \otimes \sigma_{X, X^\ast} \otimes 1_B) && \mathsf{u}^X= \cap_X; \sigma_{X, X^\ast}
\end{align*}
Furthermore, in a $\dagger$KCC, we also have the following equalities: 
\begin{align*}
    ({\mathsf{u}^X_{A,B}})^\dagger = \mathsf{tr}^X_{A,B} && ({\mathsf{u}^X})^\dagger = \mathsf{tr}^X
\end{align*}

Lastly, the final ingredient we need to properly define is the ability to conjugate maps. In a symmetric monoidal closed category, for every map $f: A \to B$ and $g: C \to D$, we obtain a map $[f,g]: [B,C] \to [A,D]$ which is interpreted precisely as pre-composing by $f$ and post-composing by $g$, so $[f,g](-) = f;-;g$. This induces a functor $[-,-]$ which contravariant in its first argument and covariant in its second argument. Now in a setting in which we have a dagger, note that for every map $f: A\to B$, we obtain the map $[f^\dagger, f]: [A,A] \to [B,B]$, which we interpret as conjugation by $f$, so $[f^\dagger, f](-) = f^\dagger; -; f$. Now in a $\dagger$KCC, for every map $f: A \to B$, we get a map $f^\ast: B^\ast \to A^\ast$, called its \textbf{transpose} \cite[Def 3.9]{coecke2016categories}, which is defined using the cups and caps, and a map $f_\ast: A \to B$, called its \textbf{conjugation} \cite[Def 3.56]{coecke2016categories}, which is defined as $f_\ast = f^{\dagger\ast} = f^{\ast\dagger}$. Then with these we get the following identities: 
\begin{align*}
    [f,g] = f^\ast \otimes g &&   [f^\dagger, f] = f_\ast \otimes f
\end{align*}
With all this notation, we are finally in a position to review abstract quantum channels in a $\dagger$KCC. 

\begin{definition}\label{Def:UCPTP} In a $\dagger$KCC, a map $\Phi: [A,A] \to [B,B]$ is said to be:
  \begin{enumerate}[(i)]
\item \label{def:CP}\textbf{Completely Positive} \cite[Cor 4.13]{selinger2007dagger} (or \textbf{CP} for short), if there exists a map $f: A \to X \otimes B$ such that:
\begin{align*}
    \Phi = [f^\dagger, f]; \mathsf{tr}^X_{B,B}
\end{align*}
\item \label{def:TP}\textbf{Trace Preserving} (or \textbf{TP} for short), if the following equality holds:
\begin{align*}
    \Phi; \mathsf{tr}^B=\mathsf{tr}^A
\end{align*}
\item \label{def:U}\textbf{Unital} (or \textbf{U} for short), if the following equality holds:
\begin{align*}
    \mathsf{u}^A;\Phi = \mathsf{u}^B
\end{align*}
\end{enumerate}
A map which is CP and TP will be called a \textbf{CPTP} map, and a map which is CPTP and unital will be called a \textbf{UCPTP} map. 
\end{definition}

\begin{remark} Notice that the definition of being unital can be defined in a mere symmetric monoidal closed category, being TP can be defined in a traced symmetric monoidal closed category, and in theory CP can be defined in a so-called dagger traced symmetric monoidal closed category ($\dagger$TSMCC). As such, one could argue that this story should be written in the setting of a $\dagger$TSMCC. However there are some currently unresolved issues with this approach. Firstly, it is unclear what exactly the coherences should be for a $\dagger$TSMCC. Secondly, while every $\dagger$KCC is a $\dagger$TSMCC, we do not currently have any good examples of a $\dagger$TSMCC which is not a $\dagger$KCC. This is also a problem for dagger symmetric monoidal closed categories \cite[Sec 9]{heunen2016monads}. Therefore we leave the exploration of $\dagger$TSMCC and quantum channels in this setting for future work. 
\end{remark}

\begin{example} In $\mathsf{FHILB}$, the object $[H,H]$ is more commonly denoted by $\mathcal{L}(H)$, the space of all $\mathbb{C}$-linear endomaps on $H$. Then, as expected, the internal trace operator $\mathsf{tr}^H: \mathcal{L}(H) \to \mathbb{C}$ sends an endomap to its trace, $\mathsf{tr}^H(\rho) = \mathsf{Tr}(\rho)$, while the internal partial trace operator $\mathsf{tr}^K_{H,H}:\mathcal{L}(K \otimes H) \to \mathcal{L}(H)$ gives the partial trace, $\mathsf{tr}^H_{K,K}(\rho) = \mathsf{Tr}_K(\rho)$. On the other hand, $\mathsf{u}^H: \mathbb{C} \to \mathcal{L}(H)$ picks out the identity map, $\mathsf{u}^H(1) = \mathsf{id}_H$. Now recall that a $\mathbb{C}$-linear map $\Phi: \mathcal{L}(H) \to \mathcal{L}(H^\prime)$ is CP in the usual sense if when tensoring with the identity of any dimension, $\mathsf{id} \otimes \Phi$ sends positive semidefinite maps to positive semidefinite maps. Alternatively, one of the well-known equivalent characterization being CP is if there exists an $\psi: H \to K \otimes H^\prime$ such that $\Phi(\rho) = \mathsf{Tr}_K\left(\psi^\dagger \circ \rho \circ \psi\right)$ \cite[Thm 2.22]{watrous2018theory}, which is precisely what Def \ref{Def:UCPTP}.(\ref{def:CP}) corresponds to. Now $\Phi$ is TP if $\mathsf{Tr}(\Phi(\rho)) = \mathsf{Tr}(\rho)$ and U if $\Phi(\mathsf{id}_H) = \mathsf{id}_{H^\prime}$, which correspond precisely to Def \ref{Def:UCPTP}.(\ref{def:TP}) and (\ref{def:U}) respectively. 
\end{example}

\begin{example} In $\mathsf{REL}$, $[X,X]= X \times X$ and so a map $\Phi: [X,X] \to [Y,Y]$ is a subset $\Phi \subseteq X \times X \times Y \times Y$. Then $\Phi$ is CP if it respects inverses \cite[Def 5.2 \& Prop 5.3]{coecke2016categories}, that is, if it satisfies (i) $(x_1,x_2,y_1,y_2) \in R \Leftrightarrow (x_2,x_1,y_2,y_1) \in R$, and (ii) $(x_1,x_2,y_1,y_2) \in R \Rightarrow (x_1,x_1, y_1, y_1) \in R$ \cite[Prop 5.3]{coecke2016categories}. On the other hand, $\Phi$ is TP if it satisfies that for all $x_1, x_2 \in X$, $x_1 = x_2  \Leftrightarrow \exists y \in Y. (x_1, x_2, y, y) \in \Phi$, and dually $\Phi$ is U if it satisfies that for all $y_1, y_2 \in Y$, $y_1 = y_2 \Leftrightarrow  \exists x \in X. (x,x, y_1, y_2)$. 
\end{example}

Now CP maps are a well-studied concept in a $\dagger$KCC and other frameworks. In particular, CP maps are closed under composition, monoidal product, and dagger \cite[Lem 4.17]{selinger2007dagger}. Therefore, we can build the category of CP maps of a $\dagger$KCC, called the CPM construction \cite[Def 4.18]{selinger2007dagger}, and it is itself a $\dagger$KCC \cite[Thm 4.20]{selinger2007dagger}. The CPM construction is arguably one of the most celebrated constructions in categorical quantum foundations, and it has been generalized to other settings \cite{coecke2016cpinf,coecke2016categories,Coecke2012,Cunningham2015,Gogioso2019}. On the other hand, while TP (resp. U) maps are closed under composition, monoidal product, and the caps and cups are TP (resp. U), the dagger of a TP (resp. U) map is not necessarily TP (resp. U). This is possibly one of the reasons why TP and U have taken less of a central role in the study of quantum channels in categorical quantum foundations. However, it is easy to see that since TP and U are dagger adjoint properties, in the sense that we get that the adjoint of a TP is TP if and only if the map was also U to begin with, which is a well known result in quantum information \cite[Thm 2.26]{watrous2018theory}. 

\begin{lemma} In a $\dagger$KCC, a map $\Phi: [A,A] \to [B,B]$ is TP (resp. U) if and only if $\Phi^\dagger: [B,B] \to [A,A]$ is U (resp. TP). Therefore, if $\Phi$ is TP (resp. U), then $\Phi^\dagger$ is TP (resp. U) if and only if $\Phi$ is U (resp. TP). 
\end{lemma}
\begin{proof} In our notation, this follows from the fact that $\left({\mathsf{u}^X}\right)^\dagger = \mathsf{tr}^X$ and $\left({\mathsf{tr}^X}\right)^\dagger = \mathsf{u}^X$. 
\end{proof}

Thus, we can conclude that the category of UCPTP maps is $\dagger$KCC (in fact a sub-$\dagger$KCC of the CPM construction). 

\begin{corollary} For a $\dagger$KCC, its category of UCPTP maps is a $\dagger$KCC. 
\end{corollary}

\section{Drazin Inverses}

One of the main result in \cite{cao2021nisq} was that the Drazin inverse of a quantum channel was TP \cite[Thm 2]{cao2021nisq}. The goal of this section is to generalize this result to the setting of $\dagger$KCC. We will show that the Drazin inverse of a TP map is again TP, and that the same result holds for U maps. For an in-depth introduction to Drazin inverses in categories, we refer the reader to \cite{cockett2024drazin}. 

\begin{definition}  \label{def:Drazin}  In a category, a \textbf{Drazin inverse} \cite[Def 2.2]{cockett2024drazin} of an endomap $f: A \to A$ is an endomap $f^D: A \to A$ such that:  
\begin{enumerate}[{\bf [$\text{D}$.1]}]
\begin{multicols}{2}
\item There is a $k \in \mathbb{N}$ such that $f^{k+1} ;f^D = f^k$
\columnbreak
\begin{multicols}{2} 
\item $f^D ;f ;f^D = f^D$
\columnbreak
\item $f ;f^D = f^D ;f$. 
\end{multicols}
\end{multicols}
\end{enumerate}
If $f$ has a Drazin inverse, we say that $f$ is \textbf{Drazin invertible} or simply \textbf{Drazin}. We call the smallest natural number $k$ such that \textbf{[D.1]} holds the \textbf{Drazin index of $f$}, which we denote by $\mathsf{ind}^D(f)$. A \textbf{Drazin category} is a category such that every map is Drazin. 
\end{definition}

Drazin inverses are \emph{unique} \cite[Prop 2.3]{cockett2024drazin}, so we may speak of \emph{the} Drazin inverse of a map and that being Drazin is a property rather than structure. 

\begin{example} $\mathsf{FHILB}$ is Drazin, so every $\mathbb{C}$-linear map $\rho: H \to H$ admits a Drazin inverse $\rho^D: H \to H$. This follows from the fact that every complex square matrix admits a Drazin inverse \cite[Chap 7]{campbell2009generalized}, where the Drazin inverse can be computed using the Jordan–Chevalley decomposition \cite[Thm 7.2.1]{campbell2009generalized}. See also \cite[Sec 2.8]{cockett2024drazin} for more details. 
\end{example}

\begin{example} $\mathsf{REL}$ is not Drazin, since in particular the successor relation $\lbrace (n, n+1) \vert~ n \in \mathbb{\mathbb{N}} \rbrace \subseteq \mathbb{N} \times \mathbb{N}$ does not have a Drazin inverse. However, if we instead consider the subcategory $\mathsf{FREL}$ of \textit{finite sets} and relations, then $\mathsf{FREL}$ is Drazin, since any category whose homsets are finite sets is Drazin \cite[Lem 2.11]{cockett2024drazin}. So if $X$ is a finite set, every subset $R \subseteq X \times X$ admits a Drazin inverse $R^D \subseteq X \times X$. 
\end{example}

Our goal is now to show that the Drazin inverse, if it exists, of a TP map is again TP. To do so, we make use of the following fact about Drazin inverses. 

\begin{proposition} \label{Drazin-commuting} \cite[Prop 3.10]{cockett2024drazin} In a category, let $f: A \to A$ and $g: B \to B$ be Drazin. For any map $k: A \to B$, if the square on the left commutes, then the square on the right commutes: 

\[ \begin{array}[c]{c} \begin{tikzcd}
	A & A \\
	B & B
	\arrow["f", from=1-1, to=1-2]
	\arrow["k"', from=1-1, to=2-1]
	\arrow["k", from=1-2, to=2-2]
	\arrow["g"', from=2-1, to=2-2]
\end{tikzcd} \end{array} \Longrightarrow \begin{array}[c]{c} \begin{tikzcd}
	A & A \\
	B & B
	\arrow["f^D", from=1-1, to=1-2]
	\arrow["k"', from=1-1, to=2-1]
	\arrow["k", from=1-2, to=2-2]
	\arrow["g^D"', from=2-1, to=2-2]
\end{tikzcd} \end{array} \]
\end{proposition}

\begin{proposition}\label{prop:DrazinTP} In a $\dagger$KCC, if $\Phi: {[A,A]} \to [A,A]$ is Drazin and TP, then $\Phi^D: {[A,A]} \to {[A,A]}$ is also TP.
\end{proposition}
\begin{proof} Identity maps $1$ are always Drazin with $1^D = 1$ \cite[Lem 3.6]{cockett2024drazin}. Therefore, by setting $g=1_I$ and $k=\mathsf{tr}^A$ in Prop \ref{Drazin-commuting}, the diagram on the left trivially commutes so the diagram on the right commutes: 
\[\begin{array}[c]{c} \begin{tikzcd}
	{[A,A]} & {[A,A]} \\
	I & I
	\arrow["\Phi", from=1-1, to=1-2]
	\arrow["{\mathsf{tr}^A}"', from=1-1, to=2-1]
	\arrow["{\mathsf{tr}^A}", from=1-2, to=2-2]
	\arrow[from=2-1, to=2-2, equal]
\end{tikzcd} \end{array} \Longrightarrow \begin{array}[c]{c}\begin{tikzcd}
	{[A,A]} & {[A,A]} \\
	I & I
	\arrow["\Phi^D", from=1-1, to=1-2]
	\arrow["{\mathsf{tr}^A}"', from=1-1, to=2-1]
	\arrow["{\mathsf{tr}^A}", from=1-2, to=2-2]
	\arrow[from=2-1, to=2-2, equal]
\end{tikzcd} \end{array} \]
Thus we conclude that the Drazin inverse is TP. 
\end{proof}

So by applying the above result to $\mathsf{FHILB}$, we obtain a novel proof of \cite[Thm 2]{cao2021nisq}, which requires none of the technical linear algebra used in \cite[App A]{cao2021nisq}. Moreover, by a similar argument, we also obtain that the Drazin inverse of a U map is again U, which to the best of our knowledge is a new observation. 

\begin{lemma}\label{lemma:DrazinU} In a $\dagger$KCC, if $\Phi: {[A,A]} \to [A,A]$ is Drazin and U, then $\Phi^D: {[A,A]} \to {[A,A]}$ is also U.
\end{lemma}

Therefore, this implies that the Drazin inverse of unital quantum channel is both TP and unital. 

\begin{remark} Notice that that Lemma \ref{lemma:DrazinU} is in fact true in a symmetric monoidal closed category, and similarly Prop \ref{prop:DrazinTP} is also true in a traced symmetric monoidal closed category. 
\end{remark}

In general, however, the Drazin inverse of CP map is not CP. Here is a simple example. 

\begin{example}\label{ex:depolar} Let $\mathbb{R}$ be the field of real numbers and let $H$ be a finite dimensional Hilbert space with dimension $d \geq 1$. Then for any $a \in \mathbb{R}$, define the depolarizing channel $\mathcal{D}_a: \mathcal{L}(H) \to \mathcal{L}(H)$ as:
\[\mathcal{D}_a(\rho) = (1-a)\cdot\rho + \frac{a \mathsf{Tr}(\rho)}{d} \cdot \mathsf{id}_H\]
Now for any $a$, $\mathcal{D}_a$ is both TP and U. However $\mathcal{D}_a$ is CP if and only if $0 \leq a \leq 1+ \frac{1}{1-d^2}$. Now if $a \neq 0$, then $\mathcal{D}_a$ is an isomorphism with inverse $\mathcal{D}_a^{-1} = \mathcal{D}_{\frac{1}{a}}$. Since the Drazin inverse of an isomorphism is its inverse \cite[Lem 3.6]{cockett2024drazin}, we get $\mathcal{D}^D_a = \mathcal{D}_{\frac{1}{a}}$, which is both TP and U as well. However, the only $a \neq 0$ for which both $a$ and $\frac{1}{a}$ are less than or equal to $1+ \frac{1}{1-d^2}$ is $1$. Therefore, for any $a \neq 1$ and $0 < a \leq 1+ \frac{1}{1-d^2}$, we have $\mathcal{D}_a$ is CP, but its Drazin inverse $\mathcal{D}^D_a = \mathcal{D}_{\frac{1}{a}}$ is not CP. 
\end{example}

Notice that the converse of either Prop \ref{prop:DrazinTP} and Lemma \ref{lemma:DrazinU} is not necessarily true, in the sense that even if the Drazin inverse is TP or U, the starting map may not be. This is because, if $f$ is Drazin, then the Drazin inverse $f^D$ is also Drazin but its Drazin inverse is not necessarily $f$, instead it's $f^{DD} = f; f^D; f$ \cite[Lem 3.11.(i)]{cockett2024drazin}. A Drazin map $f$ is the Drazin inverse of its Drazin inverse $f^D$ precisely when $\mathsf{ind}(f) \leq 1$ \cite[Lem 3.9 \& 3.11.(iii)]{cockett2024drazin}. Explicitly, an endomap $f: A \to A$ is Drazin with $\mathsf{ind}(f) \leq 1$ if and only if there is an endomap $f^D: A \to A$ such that (i) $f ;f^D ;f = f$, (ii) $f^D ;f;f^D = f^D$, and $(iii)$ $f^D ;f = f ;f^D$. In this case, $f^D$ is a \textbf{group inverse} \cite[Def 3.8]{cockett2024drazin} of $f$, and we also have that $f^{DD}=f$. 

\begin{proposition} In a $\dagger$KCC, if $\Phi: {[A,A]} \to [A,A]$ Drazin with $\mathsf{ind}^D(\Phi) \leq 1$ (or equivalently has a group inverse) then $\Phi$ is TP (resp. U) if and only if $\Phi^D: {[A,A]} \to {[A,A]}$ is also TP (resp. U).
\end{proposition}
\begin{proof} The $\Rightarrow$ direction follows immediately from Prop \ref{prop:DrazinTP} and Lemma \ref{lemma:DrazinU} respectively. For the $\Leftarrow$ direction, since $\Phi^D$ is Drazin and in this case have that $\Phi^{DD}=\Phi$, then we apply Prop \ref{prop:DrazinTP} and Lemma \ref{lemma:DrazinU} respectively to get that $\Phi$ is TP or U respectively. 
\end{proof}

\section{Dagger Drazin Inverses}

One of the limitations of the Drazin inverse is that it is defined only for endomaps $A \to A$. Of course, one would like to talk about Drazin inverses of maps of arbitrary type $A \to B$. In \cite{cockett2025dagger}, the authors achieved this by introducing the notion of dagger Drazin inverses for maps in a dagger category. 

\begin{definition}\label{def:dagger-Drazin}  In a dagger category, a \textbf{$\dagger$-Drazin inverse} \cite[Def 2.1]{cockett2025dagger} of a map $f: A \to B$ is a map of dual type $f^\partial: B \to A$ such that:  
\begin{enumerate}[{\bf [$\text{D}^{\dagger}$.1]}]
\item There is a $k \in \mathbb{N}$ such that $(f ;f^\dagger)^k ;f; f^\partial = (f ;f^\dagger)^k$ and $f^\partial ;f ;(f^\dagger; f)^k = (f^\dagger ;f)^k$

\begin{multicols}{3}
\item $f^\partial; f ;f^\partial = f^\partial$
\columnbreak
\item $(f;f^\partial)^\dagger = f;f^\partial$
\columnbreak
\item $(f^\partial ;f)^\dagger = f^\partial; f$
\end{multicols}
\end{enumerate}
If $f$ has a $\dagger$-Drazin inverse, we say that $f$ is \textbf{$\dagger$-Drazin}. 
\end{definition}

Now $\dagger$-Drazin inverses are also unique \cite[Prop 2.2]{cockett2025dagger}, so we may speak of \textit{the} $\dagger$-Drazin inverse. Moreover, $\dagger$-Drazin inverses are fundamentally linked to Drazin inverses in the following way. 

\begin{theorem}\label{thm:Drazin=dag-Drazin}\cite[Thm 3.6]{cockett2025dagger}  In a dagger category, for a map $f$, the following are equivalent: 
\begin{enumerate}[(i)]
\begin{multicols}{4}
\item $f$ is $\dagger$-Drazin
\columnbreak
\item $f^\dagger$ is $\dagger$-Drazin
\columnbreak
\item $f ;f^\dagger$ is Drazin
\columnbreak
\item $f^\dagger ;f$ is Drazin.
\end{multicols}
\end{enumerate}
Moreover, if $f$ is $\dagger$-Drazin, then the following equalities hold: 
\begin{gather*}
    f^\partial = f^\dagger ;(f ;f^\dagger)^D = (f^\dagger ;f)^D ;f^\dagger \qquad {f^\dagger}^\partial = f ;(f^\dagger ;f)^D  =  (f ;f^\dagger)^D ;f \\
    (f ;f^\dagger)^D = {f^\dagger}^\partial; f^\partial \qquad (f^\dagger ;f)^D = f^\partial ;{f^\dagger}^\partial
\end{gather*}
\end{theorem}

As such, the above characterization of $\dagger$-Drazin inverses implies that in a dagger category which is also Drazin, every map has a $\dagger$-Drazin inverse. 

Now, as expected, the $\dagger$-Drazin inverse of a CP map need not be CP, Ex \ref{ex:depolar} provides an example of such (since in this case, the Drazin inverse and $\dagger$-Drazin inverse coincide). On the other hand, we now wish to generalize the fact that taking the Drazin inverse preserves being TP and U to $\dagger$-Drazin inverses. However, we know that the Moore-Penrose inverse (which as we will soon review is a special case of a $\dagger$-Drazin inverse) of a TP map need not be TP \cite[Ex 2]{cao2021nisq}. Fortunately, as we will see below, it turns out that if a map is {\em both\/} TP and U, then its $\dagger$-Drazin inverse will also be TP and U. In order to prove this, we first generalize \cite[Prop 3.10]{cockett2024drazin} for $\dagger$-Drazin inverses. 

\begin{proposition}\label{dagger-Drazin-commuting} In a dagger category, suppose that $f: A \to B$ and $g: C \to D$ are $\dagger$-Drazin. Then if the two diagrams on the left commute, then the two diagrams on the right commute: 

\[ \begin{array}[c]{c} \begin{tikzcd}
	A & B & B & A \\
	C & D & D & C
    	\arrow["{(\textbf{a})}"{description}, draw=none, from=1-1, to=2-2]
        \arrow["{(\textbf{b})}"{description}, draw=none, from=1-3, to=2-4]
	\arrow["f", from=1-1, to=1-2]
	\arrow["h"', from=1-1, to=2-1]
	\arrow["k", from=1-2, to=2-2]
	\arrow["{f^\dagger}", from=1-3, to=1-4]
	\arrow["k"', from=1-3, to=2-3]
	\arrow["h", from=1-4, to=2-4]
	\arrow["g"', from=2-1, to=2-2]
	\arrow["{g^\dagger}"', from=2-3, to=2-4]
\end{tikzcd} \end{array} \Longrightarrow \begin{array}[c]{c} \begin{tikzcd}
	B & A & A & B \\
	D & C & C & D
    \arrow["{(\textbf{c})}"{description}, draw=none, from=1-1, to=2-2]
        \arrow["{(\textbf{d})}"{description}, draw=none, from=1-3, to=2-4]
	\arrow["{f^\partial}", from=1-1, to=1-2]
	\arrow["k"', from=1-1, to=2-1]
	\arrow["h", from=1-2, to=2-2]
	\arrow["{f^{\partial\dagger}= f^{\dagger\partial}}", from=1-3, to=1-4]
	\arrow["h"', from=1-3, to=2-3]
	\arrow["k", from=1-4, to=2-4]
	\arrow["{g^\partial}"', from=2-1, to=2-2]
	\arrow["{g^{\partial\dagger}= g^{\dagger\partial}}"', from=2-3, to=2-4]
\end{tikzcd} \end{array} \]
\vspace{-0.5cm}
\end{proposition}
\begin{proof} By pasting (\textbf{a}) and (\textbf{b}) together, we obtain the diagrams (\textbf{e}) and (\textbf{f}) below. By Thm \ref{thm:Drazin=dag-Drazin}, we know that $f;f^\dagger$, $g;g^\dagger$, $f^\dagger;f$, and $g^\dagger;g$ are all Drazin. Then by Prop \ref{Drazin-commuting}, since (\textbf{e}) and (\textbf{f}) commute, we get that (\textbf{g}) and (\textbf{h}) commute respectively: 
\[ \begin{array}[c]{c} \begin{tikzcd}
	A & A \\
	C & C
    \arrow["{(\textbf{e})}"{description}, draw=none, from=1-1, to=2-2]
	\arrow["f;f^\dagger", from=1-1, to=1-2]
	\arrow["h"', from=1-1, to=2-1]
	\arrow["h", from=1-2, to=2-2]
	\arrow["g;g^\dagger"', from=2-1, to=2-2]
\end{tikzcd} \end{array} \Longrightarrow \begin{array}[c]{c} \begin{tikzcd}
	A & A \\
	C & C
    \arrow["{(\textbf{g})}"{description}, draw=none, from=1-1, to=2-2]
	\arrow["(f;f^\dagger)^D", from=1-1, to=1-2]
	\arrow["h"', from=1-1, to=2-1]
	\arrow["h", from=1-2, to=2-2]
	\arrow["(g;g^\dagger)^D"', from=2-1, to=2-2]
\end{tikzcd} \end{array}  ,\qquad \begin{array}[c]{c} \begin{tikzcd}
	B & B \\
	D & D
    \arrow["{(\textbf{f})}"{description}, draw=none, from=1-1, to=2-2]
	\arrow["f^\dagger;f", from=1-1, to=1-2]
	\arrow["k"', from=1-1, to=2-1]
	\arrow["k", from=1-2, to=2-2]
	\arrow["g^\dagger;g"', from=2-1, to=2-2]
\end{tikzcd} \end{array} \Longrightarrow \begin{array}[c]{c} \begin{tikzcd}
	B & B \\
	D & D
         \arrow["{(\textbf{h})}"{description}, draw=none, from=1-1, to=2-2]
	\arrow["(f^\dagger;f)^D", from=1-1, to=1-2]
	\arrow["k"', from=1-1, to=2-1]
	\arrow["k", from=1-2, to=2-2]
	\arrow["(g^\dagger;g)^D"', from=2-1, to=2-2]
\end{tikzcd} \end{array}
\]
Then by pasting (\textbf{g}) with (\textbf{b)} and (\textbf{h}) with (\textbf{a}), we get that the following diagrams commute: 
\[
  \begin{tikzcd}
	A & A & B \\
	C & C & D
	\arrow["{(f;f^\dagger)^D}", from=1-1, to=1-2]
	\arrow["h"', from=1-1, to=2-1]
	\arrow["{{(\textbf{g})}}"{description}, draw=none, from=1-1, to=2-2]
	\arrow["f", from=1-2, to=1-3]
	\arrow["h", from=1-2, to=2-2]
	\arrow["{{(\textbf{a})}}"{description}, draw=none, from=1-2, to=2-3]
	\arrow["k", from=1-3, to=2-3]
	\arrow["{(g;g^\dagger)^D}"', from=2-1, to=2-2]
	\arrow["g"', from=2-2, to=2-3]
\end{tikzcd} \qquad  \begin{tikzcd}
	B & B & A \\
	D & D & B
	\arrow["{(f^\dagger;f)^D}", from=1-1, to=1-2]
	\arrow["k"', from=1-1, to=2-1]
	\arrow["{{(\textbf{h})}}"{description}, draw=none, from=1-1, to=2-2]
	\arrow["{f^\dagger}", from=1-2, to=1-3]
	\arrow["k", from=1-2, to=2-2]
	\arrow["{{(\textbf{b})}}"{description}, draw=none, from=1-2, to=2-3]
	\arrow["h", from=1-3, to=2-3]
	\arrow["{(g^\dagger;g)^D}"', from=2-1, to=2-2]
	\arrow["{g^\dagger}"', from=2-2, to=2-3]
\end{tikzcd} \]
However again by Thm \ref{thm:Drazin=dag-Drazin} we have that $f^\partial = (f^\dagger ;f)^D ;f^\dagger$, $g^\partial = (g^\dagger ;g)^D ;g^\dagger$, ${f^\dagger}^\partial = (f ;f^\dagger)^D ;f$, and ${g^\dagger}^\partial = (g ;g^\dagger)^D ;g$. Thus the above two diagrams are (\textbf{d}) and (\textbf{c}) respectively. 
\end{proof}

\begin{proposition}\label{dag-Drazin-TP} In a $\dagger$KCC, if $\Phi: {[A,A]} \to {[B,B]}$ is $\dagger$-Drazin, and both TP and U, then $\Phi^\partial: {[B,B]} \to {[A,A]}$ is TP and U. 
\end{proposition}
\begin{proof} Identity maps $1$ are always $\dagger$-Drazin with $1^\partial = 1$ \cite[Lem 2.5]{cockett2025dagger}. Therefore, by setting $g=1_I$, $h=\mathsf{tr}^A$, and $k=\mathsf{tr}^B$ in Prop \ref{dagger-Drazin-commuting}, we get that the two diagrams on the left commute trivially, which imply that the two diagrams on the right commute. 
    \[ \begin{array}[c]{c} \begin{tikzcd}
	{[A,A]} & {[B,B]} & {[B,B]} & {[A,A]} \\
	I & I & I & I
	\arrow["\Phi", from=1-1, to=1-2]
	\arrow["{\mathsf{tr}^A}"', from=1-1, to=2-1]
	\arrow["{\mathsf{tr}^B}", from=1-2, to=2-2]
	\arrow["{{\Phi^\dagger}}", from=1-3, to=1-4]
	\arrow["{\mathsf{tr}^B}"', from=1-3, to=2-3]
	\arrow["{\mathsf{tr}^A}", from=1-4, to=2-4]
	\arrow[from=2-1, to=2-2, equal]
	\arrow[from=2-3, to=2-4, equal]
\end{tikzcd} \end{array} \!\!\!\!\!\Longrightarrow \!\!\!\!\!\begin{array}[c]{c} \begin{tikzcd}
	{[B,B]} & {[A,A]} & {[A,A]} & {[B,B]} \\
	I & I & I & I
	\arrow["{\Phi^\partial}", from=1-1, to=1-2]
	\arrow["{\mathsf{tr}^B}"', from=1-1, to=2-1]
	\arrow["{\mathsf{tr}^A}", from=1-2, to=2-2]
	\arrow["{\Phi^{\dagger\partial}= \Phi^{\partial\dagger}}", from=1-3, to=1-4]
	\arrow["{\mathsf{tr}^B}"', from=1-3, to=2-3]
	\arrow["{\mathsf{tr}^A}", from=1-4, to=2-4]
	\arrow[from=2-1, to=2-2, equal]
	\arrow[from=2-3, to=2-4, equal]
\end{tikzcd} \end{array} \]
So we conclude that the $\dagger$-Drazin $\Phi^\partial$ is TP and U as desired. 
\end{proof}

We note that the converse of Prop \ref{dag-Drazin-TP} is not true. The reason is similar to why, for Drazin inverses, the converse of Prop \ref{prop:DrazinTP} is not necessarily true. Indeed, while if $f$ is $\dagger$-Drazin, its $\dagger$-Drazin inverse $f^\partial$ is also $\dagger$-Drazin.  However, $f^{\partial \partial}$ is not necessarily $f$, instead $f^{\partial \partial}= f;f^\partial;f$ \cite[Prop 2.6.(iv)]{cockett2025dagger}. This issue is resolved by considering the \textit{Moore-Penrose inverse}.  

\section{Moore-Penrose Inverses}

Arguably, the most well-known example of a generalized inverse is the famous Moore-Penrose inverse. However, as explained in \cite{cao2021nisq}, an issue with the Moore-Penrose inverse is that, unlike the Drazin inverse, the Moore-Penrose inverse of a TP map need not be TP, see \cite[Ex 2]{cao2021nisq}. In this section, we explain how this issue can be resolved by, in addition, assuming U. In fact, we will show that a map that has a Moore-Penrose inverse is TP and U if and only if its Moore-Penrose inverse is TP and U. For a detailed introduction to Moore-Penrose inverses in dagger categories, we refer the reader to see \cite{Cockett2023Moore}.

\newpage

\begin{definition} In a dagger category, a \textbf{Moore-Penrose inverse} \cite[Def 2.3]{Cockett2023Moore} of a map $f: A \to B$ is a map of dual type $f^\circ: B \to A$ such that:
\begin{enumerate}[{\bf [$\text{MP}$.1]}]
\begin{multicols}{2}
    \item $f ;f^\circ ;f = f$
    \item $f^\circ ;f ;f^\circ = f^\circ$
\columnbreak
\item $(f;f^\circ)^\dagger = f;f^\circ$
    \item $(f^\circ ;f)^\dagger = f^\circ ;f$
\end{multicols}
\end{enumerate}
If $f$ has a Moore-Penrose inverse, we say that $f$ is \textbf{Moore-Penrose} (or \textbf{MP} for short). A \textbf{Moore-Penrose dagger category} is a dagger category such that every map is MP. 
\end{definition}

Moore-Penrose inverses are \emph{unique} \cite[Lem 2.4]{Cockett2023Moore}, so again we may speak of \emph{the} Moore-Penrose inverse of a map. Moreover, if $f$ is MP, then its MP inverse $f^\circ$ is also MP as $f^{\circ\circ} = f$ \cite[Lem 2.5.(i)]{Cockett2023Moore} (which, as we have seen, is not necessarily true for other kinds of generalized inverses). 

\begin{example} $\mathsf{FHILB}$ is Moore-Penrose, so every $\mathbb{C}$-linear map $f: H \to K$ admits a MP inverse $f^\circ: K \to H$ which can be constructed using singular value decomposition \cite[Ex 2.13]{Cockett2023Moore}. 
\end{example}

\begin{example} $\mathsf{FREL}$ is not Moore-Penrose but we do know which maps are MP. Recall that a relation $R \subseteq X \times Y$ is \textit{difunctional} if $(x,b),(a,b),(a,y) \in R$ implies that $(x,y) \in R$ (which from a dagger category perspective means that $R$ is a partial isometry). Then a relation $R \subseteq X \times Y$ is MP if and only if it is difunctional, and in this case the MP inverse is the inverse relation, $R^\circ = R^\dagger$ \cite[Ex 2.14]{Cockett2023Moore}. 
\end{example}

It turns out that an MP inverse is a special case of the $\dagger$-Drazin inverse. In fact, MP maps are precisely the $\dagger$-Drazin maps that are the $\dagger$-Drazin inverse of their $\dagger$-Drazin inverse. 

\begin{theorem}\label{thm:MP=a-dag-Drazin} \cite[Thm 4.2]{cockett2025dagger} In a dagger category, a map $f$ is MP if and only if $f$ is $\dagger$-Drazin and $f^{\partial \partial} = f$. In this case, the MP inverse and $\dagger$-Drazin inverse coincide, $f^\partial = f^\circ$. 
\end{theorem}

With this in hand, we can easily extend Prop \ref{dagger-Drazin-commuting} to MP inverses, and show that MP inverses preserve being both TP and U, as desired. 

\begin{proposition}\label{MP-commuting} Suppose that $f: A \to B$ and $g: C \to D$ are MP. Then, if the two diagrams on the left commute, then the two diagrams on the right commute: 

\[ \begin{array}[c]{c} \begin{tikzcd}
	A & B & B & A \\
	C & D & D & C
	\arrow["f", from=1-1, to=1-2]
	\arrow["h"', from=1-1, to=2-1]
	\arrow["k", from=1-2, to=2-2]
	\arrow["{f^\dagger}", from=1-3, to=1-4]
	\arrow["k"', from=1-3, to=2-3]
	\arrow["h", from=1-4, to=2-4]
	\arrow["g"', from=2-1, to=2-2]
	\arrow["{g^\dagger}"', from=2-3, to=2-4]
\end{tikzcd} \end{array} \Longrightarrow \begin{array}[c]{c} \begin{tikzcd}
	B & A & A & B \\
	D & C & C & D
	\arrow["{f^\circ}", from=1-1, to=1-2]
	\arrow["k"', from=1-1, to=2-1]
	\arrow["h", from=1-2, to=2-2]
	\arrow["{f^{\circ\dagger}= f^{\dagger\circ}}", from=1-3, to=1-4]
	\arrow["h"', from=1-3, to=2-3]
	\arrow["k", from=1-4, to=2-4]
	\arrow["{g^\circ}"', from=2-1, to=2-2]
	\arrow["{g^{\circ\dagger}= g^{\dagger\circ}}"', from=2-3, to=2-4]
\end{tikzcd} \end{array} \]
\end{proposition}
\begin{proof} Since $f$ and $g$ are MP, they are both $\dagger$-Drazin with $f^\partial= f^\circ$ and $g^\partial=g^\circ$. Then by applying Prop \ref{dagger-Drazin-commuting}, we obtain the desired result. 
\end{proof}

\begin{proposition}\label{MP-TP} In a $\dagger$KCC, if $\Phi: {[A,A]} \to {[B,B]}$ is MP, then $\Phi$ is TP and U if and only if $\Phi^\circ: {[B,B]} \to {[A,A]}$ is TP and U. 
\end{proposition}
\begin{proof} For the $\Rightarrow$ direction, since $\Phi$ is MP, it is also $\dagger$-Drazin and $\Phi^\partial=\Phi^\circ$. By applying Prop \ref{dag-Drazin-TP} we get that $\Phi^\circ$ is TP and U. For the $\Leftarrow$ direction, since $\Phi$ is MP, $\Phi^{\circ}$ is also MP and $\Phi^{\circ\circ} = \Phi$. So by applying the $\Rightarrow$ direction that we just proved with $\Phi^\circ$, we get $\Phi$ is TP and U as desired. 
\end{proof}

Therefore this tells us that for unital quantum channels, the MP inverse is both TP and unital. 

\section{(U)CPTP Maps whose Generalized Inverses are (U)CPTP}

As we have established, while generalized inverses preserve being TP and/or U, they often do not preserve being CP. Thus, a natural question to ask is whether there are any quantum channels whose generalized inverse are again quantum channels. The goal of this section is to provide some examples of (U)CPTP maps whose Drazin, $\dagger$-Drazin, or Moore-Penrose inverses are again (U)CPTP. Our first source of examples are the so called pure channels. 

\begin{definition} In a $\dagger$KCC, a map $\Phi: [A,A] \to [B,B]$ is \textbf{pure} if there exists a map $f: A \to B$ such that:
\[\Phi = [f^\dagger,f]\]
\end{definition}

The properties of $f$ determine the properties of the induced map $[f^\dagger,f]$. Recall that in a dagger category, a map $f: A \to B$ is an \textbf{isometry} if $f;f^\dagger = 1_A$, a \textbf{coisometry} if $f;f^\dagger= 1_B$, and \textbf{unitary} if it both an isometry and a coisometry, which particular means it is an isomorphism whose inverse is its adjoint, $f^{-1} = f^\dagger$. Similarly, we say that map $\Phi: [A,A] \to [B,B]$ is a \textbf{pure isometry (resp. unitary)} if there exists a isometry (resp. unitary) map $f: A \to B$ such that $\Phi = [f^\dagger,f]$. 

\begin{lemma} 
\label{Lemma: Properties of [fdagger, f]}
In a $\dagger$KCC, for any map $f: A \to B$, 
   \begin{enumerate}[(i)]
\item $[f^\dagger,f]$ is CP \cite[Lem 4.17.(d)]{selinger2007dagger};
\item  $[f^\dagger,f]$ is TP if and only if $f$ is an isometry;
\item $[f^\dagger,f]$ is U if and only if $f$ is a coisometry;
\item $[f^\dagger,f]$ is UCPTP if and only if $f$ is unitary. 
\end{enumerate}
Therefore, a pure isometry map is CPTP and a pure unitary map is UCPTP. 
\end{lemma}
\begin{proof} (i) is immediate by construction. For (ii)-(iv), these follow from the fact that since $\mathsf{tr}^X$ and $\mathsf{u}^X$ are extranatural, we have that:
\begin{align*}
    [f^\dagger;f];\mathsf{tr}^B= [1_A, f;f^\dagger];\mathsf{tr}^A = [f^\dagger;f, 1_B];\mathsf{tr}^B && \mathsf{u}^A;[f^\dagger;f]= \mathsf{u}^A;[1_A, f;f^\dagger] = \mathsf{u}^B;[f^\dagger;f, 1_B]
\end{align*}
\end{proof}

\begin{lemma}
\label{Lemma: drazin inverse of pure maps}
In a $\dagger$KCC, if a map $f$ is Drazin or $\dagger$-Drazin or MP, then $[f^\dagger, f]$ is Drazin where $[f^\dagger, f]^D = [{f^D}^\dagger, f^D]$ or $\dagger$-Drazin where $[f^\dagger, f]^\partial = [{f^\partial}^\dagger, f^\partial]$ or MP where $[f^\dagger, f]^\circ= [{f^\circ}^\dagger, f^\circ]$ respectively. In particular, if $f$ is unitary, then $[f^\dagger, f]$ is Drazin, $\dagger$-Drazin, and MP where $[f^\dagger, f]^D = [f^\dagger, f]^\partial = [f^\dagger, f]^\circ= [f, f^\dagger]$. 
\end{lemma}
\begin{proof} This easily follows from the fact that Drazin inverses are absolute, meaning any functors preserves Drazin inverses \cite[Prop 3.13]{cockett2024drazin}, that is, if $F: \mathbb{X} \to \mathbb{Y}$ is a functor, if $f: A \to A$ is Drazin in $\mathbb{X}$, then $F(f): F(A) \to F(A)$ is Drazin in $\mathbb{Y}$ with $F(f)^D = F(f^D)$. It is easy to see that $\dagger$-Drazin inverses or MP inverses are dagger absolute, that is, any dagger functor preserves $\dagger$-Drazin inverses or MP inverses. Explicitly, recall that a dagger functor is a functor $F: \mathbb{X} \to \mathbb{Y}$ between dagger categories which commutes with the dagger, that is, $F(f^\dagger) = F(f)^\dagger$. Then it easy to show that if $f: A \to B$ is $\dagger$-Drazin (resp. MP) in $\mathbb{X}$, then $F(f): F(A) \to F(B)$ is also $\dagger$-Drazin (resp. MP) in $\mathbb{Y}$ where $F(f)^\partial = F(f^\partial)$ (resp. $F(f)^\circ = F(f^\circ)$). Now let $\mathbb{X}$ be a $\dagger$KCC and consider the endofunctor $F: \mathbb{X} \to \mathbb{X}$ defined on objects as $F(A) = [A,A]$ and on maps as $F(f) = [f^\dagger, f]$. Moreover, this $F$ is a dagger functor. Therefore, $F$ preserves Drazin inverses, $\dagger$-Drazin inverses, and MP inverse, and so the desired identities hold.
\end{proof}

A well-known result about quantum channels is that they admit a Kraus representation \cite[Cor 2.27]{watrous2018theory}, meaning that they are all sums of pure isometric channels. To consider quantum channels in our setting, we will have to add additive structure to our $\dagger$KCC. 

So for a category with biproducts \cite[Sec 6.3]{selinger2010survey}, we denote the zero object by $\mathsf{0}$, and the biproduct of a family of objects by $A_1 \oplus \hdots \oplus A_n$ with projection maps $\pi_j: A_1 \oplus \hdots \oplus A_n \to A_j$ and injection maps $\iota_j: A_j \to A_1 \oplus \hdots \oplus A_n$. Recall that every category with biproducts is also enriched over commutative monoids, which essentially means we have zero maps, $0$, and can add parallel maps, $f+g$. Moreover, composition is a monoid morphism so that:
\begin{align*}
    h;0;k =0 && h;(f+g); k = h;f;k + h;g;k
\end{align*}
Furthermore, recall that the injections and projections also satisfy that:
\begin{align*}
    \iota_i ;\pi_j = \left\{ \begin{array}[c]{ll} 1_{A_i} & \mbox{when $i=j$} \\ 0 & \mbox{when $i \not= j$} \end{array} \right. && \sum\limits^n_{i=1} \pi_j; \iota_j = 1_{A_1 \oplus \hdots \oplus A_n}
\end{align*}

Now a dagger category has dagger biproducts \cite[Sec 7.6]{selinger2010survey} if it has biproducts such that the projection and injections are adjoints of each other, so:
\begin{align*}
    \pi_i^\dagger=\iota_i && \iota^\dagger_i = \pi_i
\end{align*}
In this setting, we also get that the dagger is also a monoid morphism, so:
\begin{align*}
    (f+g)^\dagger = f^\dagger + g^\dagger && 0^\dagger = 0
\end{align*}
Then by a \textbf{dagger additive compact closed category} (or \textbf{$\dagger$AKCC} for short), we mean a $\dagger$KCC with dagger biproducts. It is worth noting that in this setting, it comes for free that the monoidal product preserves both the sum and the biproduct in the expected way.

\begin{definition} In a $\dagger$AKCC, a map $\Phi: [A,A] \to [B,B]$ is said to be \textbf{mixed pure} if there exits a family of map $\lbrace f_i: A \to B \rbrace^n_{i=1}$ such that:
\[\Phi = \sum\limits_{i=1}^n [f_i^\dagger,f_i]\] 
\end{definition}

\begin{lemma} 
\label{Lemma: pure-sum}
In a $\dagger$AKCC, let $\lbrace f_i: A \to B \rbrace^n_{i=1}$ a family of maps. Then: 
   \begin{enumerate}[(i)]
\item $\sum\limits_{i=1}^n [f_i^\dagger,f_i]$ is CP
\item $\sum\limits_{i=1}^n [f_i^\dagger,f_i]$ is TP if and only if $\sum\limits_{i=1}^n f_i;f_i^\dagger = 1_A$
\item $\sum\limits_{i=1}^n [f_i^\dagger,f_i]$ is U if and only if $\sum\limits_{i=1}^n f_i^\dagger;f_i = 1_B$. 
\end{enumerate}

\end{lemma}
\begin{proof} For (i), this follows from the fact that sum of CP maps is again a CP map \cite[Lem 5.2]{selinger2007dagger}. The other two statements are easy to check.  
\end{proof}

As with the usual inverse, computing the generalized inverse of a sum can be quite complicated \cite{cline1965representations}, and in general the generalized inverse of the sum need not be the sum of the generalized inverses. However, if we assume a minimal orthogonal condition on our family of maps, then we can show that in this case, the generalized inverse of their sum is the sum of their generalized inverse. Thus we get the following result in a commutative monoid enriched category. 

\begin{lemma}\label{Drazin-sum} In a category enriched over commutative monoids, let $\lbrace f_i: A \to A \rbrace^n_{i=1}$ be a family of Drazin endomaps such that $f_i ;f_j=0$ for all $i \neq j$. Then $\sum \limits^n_{i=1} f_i$ is Drazin where:
\[\left(\sum \limits^n_{i=1} f_i \right)^D = \sum \limits^n_{i=1} f_i^D\]
\end{lemma}
\begin{proof} We need to show that $\sum \limits^n_{i=1} f_i^D$ satisfies \textbf{[D.1]-[D.3]}. Let us start with \textbf{[D.3]}. Since $f_i;f_j = 0 = f_j;f_i$, by Prop \ref{Drazin-commuting} and that $0$ is Drazin with $0^D=0$ \cite[Lem 6.10]{cockett2024drazin}, it follows that $f_i^D;f_j = 0 = f_j^D;f_i$. So we get that: 
\[\left(\sum \limits^n_{i=1} f_i \right) ; \left(\sum \limits^n_{i=1} f_i^D \right) = \sum \limits^n_{i=1} f_i;f_i^D = \sum \limits^n_{i=1} f_i^D ;f_i = \left(\sum \limits^n_{i=1} f_i^D \right) ; \left(\sum \limits^n_{i=1} f_i \right)\]
So \textbf{[D.3]} holds. Moreover, note that $f_j; f_i; f_i^D =0$, then it follows that:
\[\left(\sum \limits^n_{i=1} f_i^D \right);\left(\sum \limits^n_{i=1} f_i \right) ; \left(\sum \limits^n_{i=1} f_i^D \right) = \sum \limits^n_{i=1} f_i^D ;f_i;f_i^D =\sum \limits^n_{i=1} f_i^D\]
So \textbf{[D.2]} holds. Now for \textbf{[D.1]}, first note that in our case we have that $\left(\sum \limits^n_{i=1} f_i \right)^{k+1} = \sum\limits^{n}_{i=1} f_i^{k+1}$. Then set $k=\max\lbrace \mathsf{ind}^D(f_i) \rbrace$. In particular this means that for all $1 \leq i \leq n$, we hae that $f_i^{k+1} ;f^D_i = f^{k}_i$ \cite[Lem 2.3.(i)]{cockett2024drazin}. Therefore, since $f_i^{k+1};f_j^D =0$, we have that: 
\[\left(\sum \limits^n_{i=1} f_i \right)^{k+1}; \left(\sum \limits^n_{i=1} f_i^D \right) = \sum \limits_{i=1}^n f^{k+1}_i;f_i^D = \sum \limits_{i=1}^n f^{k}_i =\left(\sum \limits^n_{i=1} f_i \right)^{k}\] 
So \textbf{[D.1]} holds. Therefore $\left(\sum \limits^n_{i=1} f_i \right)^D = \sum \limits^n_{i=1} f_i^D$ is the Drazin inverse of $\sum \limits^n_{i=1} f_i$ as desired. 
\end{proof}

With this in hand, we can build a (U)CPTP map whose Drazin inverse is also (U)CPTP. 

\begin{corollary}\label{cor:draz-phi-sum} In a $\dagger$AKCC, let $\lbrace f_i: A \to A \rbrace^n_{i=1}$ be a family of Drazin endomaps such that $f_i ;f_j=0$ for all $i \neq j$, and let $\Phi = \sum\limits_{i=1}^n [f_i^\dagger,f_i]$. Then $\Phi$ is CP and Drazin where 
\[\Phi^D = \sum\limits_{i=1}^n [f_i^{D\dagger},f_i^D]\] 
Thus $\Phi^D$ is also CP. Furthermore, if $\sum\limits_{i=1}^n f_i;f_i^\dagger = 1_A$ (resp. $\sum\limits_{i=1}^n f_i^\dagger;f_i = 1_A$), then both $\Phi$ and $\Phi^D$ are TP (resp. U). 
\end{corollary}
\begin{proof} By Lem \ref{Lemma: pure-sum}.$(i)$, $\Phi = \sum\limits_{i=1}^n [f_i^\dagger,f_i]$ is CP. Now since $f_i;f_j=0$ for all $i\neq j$, it follows that $[f^\dagger_i,f_i];[f^\dagger_j,f_j] =0$ as well. Then by combining Lem \ref{Lemma: drazin inverse of pure maps} and Lem \ref{Drazin-sum}, we get that $\Phi^D = \sum \limits^n_{i=1} [f_i^{\dagger D}, f_i^D]$. Hence, by Lem \ref{Lemma: pure-sum}.$(i)$, $\Phi^D$ is CP as well. It follows from Lem \ref{Lemma: pure-sum}.$(ii)$ (resp. Lem \ref{Lemma: pure-sum}.$(iii)$) that if $\sum\limits_{i=1}^n f_i;f_i^\dagger = 1_A$ (resp. $\sum\limits_{i=1}^n f_i^\dagger;f_i = 1_A$), then both $\Phi$ and $\Phi^D$ are TP (resp. U).
\end{proof}

Let us now extend Lem \ref{Drazin-sum} to $\dagger$-Drazin inverses and to MP inverse. 

\begin{lemma}\label{lem:dag-draz-sum} In a dagger category enriched over commutative monoids, let $\lbrace f_i: A \to B \rbrace^n_{i=1}$ be a family of $\dagger$-Drazin maps such that $f_i; f^\dagger_j=0$ for all $i \neq j$. Then $\sum \limits^n_{i=1} f_i$ is $\dagger$-Drazin where:
\[\left(\sum \limits^n_{i=1} f_i \right)^\partial = \sum \limits^n_{i=1} f_i^\partial\]
\end{lemma}
\begin{proof} First observe that since $f_i; f^\dagger_j=0$ for all $i \neq j$, we get that:
\[\left( \sum \limits^n_{i=1} f_i\right); \left( \sum \limits^n_{i=1} f_i\right)^\dagger= \sum \limits^n_{i=1} f_i;f_i^\dagger\]
Now since each $f_i$ is $\dagger$-Drazin, we have that $f_i;f_i^\dagger$ is Drazin by Thm \ref{thm:Drazin=dag-Drazin}. Moreover, we have that since $f_i; f^\dagger_j=0$ for all $i \neq j$, we also get that $f_i;f_i^\dagger;f_j;f_j^\dagger=0$ for all $i \neq j$. Thus by Lem \ref{Drazin-sum}, we get that:
\[\left( \left( \sum \limits^n_{i=1} f_i\right); \left( \sum \limits^n_{i=1} f_i\right)^\dagger \right)^D = \left(\sum \limits^n_{i=1} f_i;f_i^\dagger\right)^D = \sum \limits^n_{i=1} (f_i;f_i^\dagger)^D\]
Therefore we know that $\sum \limits^n_{i=1} f_i$ is also $\dagger$-Drazin via the formula from Thm \ref{thm:Drazin=dag-Drazin}, which is easily seen to provide $\left(\sum \limits^n_{i=1} f_i \right)^\partial = \sum \limits^n_{i=1} f_i^\partial$ as desired. 
\end{proof}

\begin{corollary} In a dagger category enriched over commutative monoids, let $\lbrace f_i: A \to B \rbrace^n_{i=1}$ be a family of MP maps such that $f_i; f^\dagger_j=0$ for all $i \neq j$. Then $\sum \limits^n_{i=1} f_i$ is MP where:
\[\left(\sum \limits^n_{i=1} f_i \right)^\circ = \sum \limits^n_{i=1} f_i^\circ\]
\end{corollary} 
\begin{proof} This follows immediately from the fact that MP inverses are special cases of $\dagger$-Drazin inverses and applying Lem \ref{lem:dag-draz-sum}.   
\end{proof}

\begin{corollary}\label{cor:dag-draz-phi-sum} In a $\dagger$AKCC, let $\lbrace f_i: A \to B \rbrace^n_{i=1}$ be a family of $\dagger$-Drazin (resp. MP) maps such that $f_i ;f_j^\dagger=0$ for all $i \neq j$, and let $\Phi = \sum\limits_{i=1}^n [f_i^\dagger,f_i]$. Then $\Phi$ is CP and $\dagger$-Drazin (resp. MP), where:
\begin{align*}
    \Phi^\partial = \sum\limits_{i=1}^n [f_i^{\partial\dagger},f_i^\partial] && \left(\text{resp. } \Phi^\circ = \sum\limits_{i=1}^n [f_i^{\circ\dagger},f_i^\circ] \right)
\end{align*}
Thus $\Phi^\partial$ (resp. $\Phi^\circ$) is CP. Furthermore, if $\sum\limits_{i=1}^n f_i;f_i^\dagger = 1_A$ (resp. $\sum\limits_{i=1}^n f_i^\dagger;f_i = 1_B$), then both $\Phi$ and $\Phi^\partial$ (resp. $\Phi^\circ$) are TP (resp. U). 
\end{corollary}
\begin{proof} By Lem \ref{Lemma: pure-sum}.$(i)$, $\Phi = \sum\limits_{i=1}^n [f_i^\dagger,f_i]$ is CP. Morever, since $f_i;f^\dagger_j=0$ for all $i\neq j$, it follows that $[f^\dagger_i,f_i];[f^\dagger_j,f_j]^\dagger =0$ as well. Thus by combining Lem \ref{Lemma: drazin inverse of pure maps} and Lem \ref{lem:dag-draz-sum}, we get that $\Phi^\partial = \sum\limits_{i=1}^n [f_i^{\partial\dagger},f_i^\partial]$ (resp. $\Phi^\circ = \sum\limits_{i=1}^n [f_i^{\circ\dagger},f_i^\circ]$). Hence, by Lem \ref{Lemma: pure-sum}.$(i)$, $\Phi^\partial$ (resp. $\Phi^\circ$) is CP as well. It follows from Lem \ref{Lemma: pure-sum}.$(ii)$ (resp. Lem \ref{Lemma: pure-sum}.$(iii)$) that if $\sum\limits_{i=1}^n f_i;f_i^\dagger = 1_A$ (resp. $\sum\limits_{i=1}^n f_i^\dagger;f_i = 1_A$), then both $\Phi$ and $\Phi^\partial$ (or $\Phi^\circ$) are TP (resp. U).
\end{proof}

It is easy enough to find examples that satisfy the necessary conditions in Cor \ref{cor:draz-phi-sum} and Cor \ref{cor:dag-draz-phi-sum}. Indeed, given a biproduct $\bigoplus\limits^n_{i=1} A_i$, consider the maps $e_i: \bigoplus\limits^n_{i=1} A_i \to \bigoplus\limits^n_{i=1} A_i$ defined as the composite $e_i=\pi_i;\iota_i$. Each $e_i$ is self-adjoint and idempotent, and moreover, since we have dagger biproducts we have that $e_i;e_j=0$ when $i \neq j$ and also $\sum^n_{i=1} e_i^2 = 1_{\bigoplus\limits^n_{i=1} A_i}$. Therefore $\sum\limits^n_{i=1} \left[ e_i^\dagger, e_i \right]: \left[\bigoplus\limits^n_{i=1} A_i, \bigoplus\limits^n_{i=1} A_i\right] \to \left[\bigoplus\limits^n_{i=1} A_i, \bigoplus\limits^n_{i=1} A_i\right]$ is a UCPTP, which turns out to be its own Drazin/$\dagger$-Drazin/MP inverse. So this gives us an example of an abstract quantum channel, whose generalized inverse is a quantum channel. 

\nocite{*}
\bibliographystyle{eptcs}
\bibliography{qpl2026}
\end{document}